\documentclass[10pt,a4paper]{article}
\usepackage[hscale=0.7,scale=0.75]{geometry}

\newcommand{\citex}[1]{\citeauthor{#1}~\shortcite{#1}}

\makeatletter
\def\leftcite{\@up[}\def\rightcite{\@up]}

\def\cite{\def\citeauthoryear##1##2{\def\@thisauthor{##1}%
    \ifx \@lastauthor \@thisauthor \relax \else##1, \fi ##2}\@icite}
\def\shortcite{\def\citeauthoryear##1##2{##2}\@icite}
  
\def\citeauthor{\def\citeauthoryear##1##2{##1}\@nbcite}
\def\citeyear{\def\citeauthoryear##1##2{##2}\@nbcite}

\def\@icite{\leavevmode\def\@citeseppen{-1000}%
  \def\@cite##1##2{\leftcite\nobreak\hskip 0in{##1\if@tempswa , ##2\fi}\rightcite}%
  \@ifnextchar [{\@tempswatrue\@citex}{\@tempswafalse\@citex[]}}
\def\@nbcite{\leavevmode\def\@citeseppen{1000}%
  \def\@cite##1##2{{##1\if@tempswa , ##2\fi}}%
  \@ifnextchar [{\@tempswatrue\@citex}{\@tempswafalse\@citex[]}}

\def\@citex[#1]#2{%
  \def\@lastauthor{}\def\@citea{}%
  \@cite{\@for\@citeb:=#2\do
    {\@citea\def\@citea{;\penalty\@citeseppen\ }%
      \if@filesw\immediate\write\@auxout{\string\citation{\@citeb}}\fi
      \@ifundefined{b@\@citeb}{\def\@thisauthor{}{\bf ?}\@warning
        {Citation `\@citeb' on page \thepage \space undefined}}%
      {\csname b@\@citeb\endcsname}\let\@lastauthor\@thisauthor}}{#1}}

\def\@biblabel#1{\def\citeauthoryear##1##2{##1, ##2}\@up{[}#1\@up{]}\hfill}

\def\@up#1{\leavevmode\raise.2ex\hbox{#1}}
\makeatother

\usepackage{amssymb}
\usepackage{latexsym}
\usepackage{times}
\usepackage{amsmath}
\usepackage{microtype}
\usepackage{url}
\usepackage{xspace}

\newtheorem{definition}{Definition}
\newtheorem{theorem}{Theorem}
\newtheorem{proposition}{Proposition}
\newtheorem{lemma}{Lemma}
\newtheorem{corollary}{Corollary}
 
\def\qed{\hspace*{\fill}{\rule[-0.5mm]{1.5mm}{3mm}}}

\def\endproof{\ifhmode\nobreak\qed\par\fi\medskip}

\def\AND     { \wedge                 }

\renewcommand{\mid}{\;{|}\;}%

\usepackage{paralist}

\newcommand{\AS}{\ensuremath{\mathit{AS}}}

\newcommand{\Mod}{\ensuremath{\mathit{Mod}}}

\newcommand{\pmm}[3]{\ensuremath{{#1}|^{#2}_{#3}}}

\newcommand{\Pol}{{\rm P}}

\newcommand{\NP}{\mbox{\rm NP}\xspace}

\newcommand{\coNP}{\mbox{\rm coNP}\xspace}

\newcommand{\PiP}[1]{{\Pi}_{#1}^{P}}

\newcommand{\naf}{\ensuremath{\neg}}
\newcommand{\head}[1]{H(#1)}

\newcommand{\body}[1]{B(#1)}
\newcommand{\bodyp}[1]{B^+(#1)}
\newcommand{\bodyn}[1]{B^-(#1)}
\newcommand{\la}{\leftarrow}

\newcommand{\mse}{\models_{\mathit{SE}}}

\newcommand{\equivu}{\equiv_{u}}
\newcommand{\equivs}{\equiv_{s}}

\newcommand{\SEQ}{\mathit{SE}}

\newcommand{\UE}{\ensuremath{\mathit{UE}}}

\newcommand{\hX}{{\hat{X}}}
\newcommand{\hY}{{\hat{Y}}}

\renewcommand{\S}{{\cal S}}

\newcommand{\X}{{\cal X}}
\newcommand{\Y}{{\cal Y}}
\newcommand{\Z}{{\cal Z}}
\newcommand{\U}{{\cal U}}

\newcommand{\nop}[1]{}
\newcommand{\Comment}[1]{}

\newcommand{\commadots}[0]{,\ldots ,}

\newcommand{\oln}[1]{\overline{#1}}
\newcommand{\wh}[1]{\widehat{#1}}

\newcommand{\SB}{\{\,}%
\newcommand{\SM}{\;{|}\;}%
\newcommand{\SE}{\,\}}%

\newcommand{\Card}[1]{|#1|}
\newcommand{\CCard}[1]{\|#1\|}

\let\phi=\varphi
\let\epsilon=\varepsilon

\let\phi=\varphi

\newcommand{\at}{\text{\normalfont at}}
\newcommand{\At}{\mathit{At}}

\newcommand{\rsep}{;\;}

\newcommand{\F}[1]{F^{\mathit{#1}}}

\newcommand{\problemn}[1]{{\scshape #1}}
\newcommand{\SAT}{\problemn{Sat}\xspace}
\newcommand{\UNSAT}{\problemn{UnSat}\xspace}

\title{%
  Dual-normal Logic Programs -- the Forgotten Class%
  \footnote{This is the author's self-archived copy including detailed
    proofs. To appear in Theory and Practice of Logic Programming
    (TPLP), Proceedings of the 31st International Conference on Logic
    Programming (ICLP 2015).} %
  \thanks{This work has been funded by the Austrian Science Fund (FWF)
    through projects Y698 and P25518.} %
}

\author{%
  Johannes K. Fichte\\
  TU Wien, Austria\\
  University of Potsdam, Germany\\
  (e-mail: {fichte@kr.tuwien.ac.at})
  \and  Miros\l{}aw Truszczy\'nski\\
  University of Kentucky, Lexington, KY, USA\\
  (e-mail: {mirek@cs.engr.uky.edu})
  \and Stefan Woltran\\
  TU Wien, Austria\\
  (e-mail: {woltran@dbai.tuwien.ac.at})
}

\begin{document}
\maketitle

\begin{abstract}
Disjunctive Answer Set Programming is a powerful
declarative programming paradigm with complexity
beyond \NP. Identifying classes of programs for
which the consistency problem is in \NP is
of interest from the theoretical standpoint and can potentially lead to
improvements in the design of answer set programming solvers.
One of such classes consists of \emph{dual-normal programs},
where the number of positive body atoms in proper rules
is at most one. Unlike other classes of programs,
dual-normal programs have received little attention so far.
In this paper we study this class. We
relate dual-normal programs to propositional theories and
to normal programs by presenting several inter-translations.
With the translation from dual-normal to normal programs at hand, 
we introduce the novel class of \emph{body-cycle} free programs,
which are in many respects dual to head-cycle free programs.
We establish the expressive power of dual-normal programs in terms
of SE- and UE-models, and compare them to normal programs. We also 
discuss the complexity of deciding whether dual-normal
programs are strongly and uniformly equivalent.
\end{abstract}

\section{Introduction}

Disjunctive Answer Set Programming (ASP)~\cite{brew-etal-11-asp} is a
vibrant area of AI providing a declarative formalism for solving hard
computational problems. Thanks to the power of modern ASP
technology~\cite{2012Gebser}, ASP was successfully used in many
application areas, including product configuration~\cite{SoininenN99},
decision support for space shuttle flight
controllers~\cite{NogueiraBGWB01,BalducciniGN06}, team
scheduling~\cite{RiccaGAMLIL12}, and
bio-informatics~\cite{GuziolowskiVETCSS13}.

With its main decision problems located at the second level of the
polynomial hierarchy, full disjunctive ASP is necessarily
computationally involved. But some fragments of ASP have lower
complexity. Two prominent examples are the class of \emph{normal}
programs and the class of \emph{head-cycle free} (HCF)
programs~\cite{Ben-EliyahuD94}. In each case, the problem of the
existence of an answer set is \NP-complete. Identifying and
understanding such fragments is of theoretical importance and can also
help to make ASP solvers more efficient. A solver can detect whether a
program is from an easier class (e.g., is normal or head-cycle free)
and, if so, use a dedicated more lightweight machinery to process it.

HCF programs are defined by a \emph{global} condition taking into
account all rules in a program. On the other hand, interesting classes
of programs can also be obtained by imposing conditions on
\emph{individual} rules. Examples include the classes of Horn, normal,
negation-free, and purely negative programs. For instance, Horn
programs consist of rules with at most one atom in the head and no
negated atoms in the body, and purely negative programs consist of
rules with no atoms in the positive body. A general schema to define
classes of programs in terms of the numbers of atoms in the head and
in the positive and negative bodies of their rules was proposed
by~\citex{Truszczynski11}. In the resulting space of classes of
programs, the complexity of the \emph{consistency} problem (that is,
the problem of the existence of an answer set) ranges from \Pol~to
\NP-complete to $\Sigma_2^P$-complete. The three main classes of
programs in that space that fall into the \NP-complete category are
the classes of normal and negation-free programs (possibly with
constraints), mentioned above, and the class of programs whose
non-constraint rules have at most one positive atom in the
body~\cite{Truszczynski11}. While the former two classes have been
thoroughly investigated, the third class has received little attention
so far. In particular, the paper by \citex{Truszczynski11} only
identified the class and established the complexity of the main
reasoning tasks (deciding the consistency, and skeptical and credulous
reasoning).

In this paper, we study this ``forgotten'' class in more detail.  We
call its programs \emph{dual-normal}, since the reducts of their
non-constraint part are \emph{dual}-Horn. In fact, this is the reason
why for dual-normal programs the consistency problem is in \NP.
Lower complexity is not the only reason why dual-normal programs are
of interest. Let us consider a slight modification of the celebrated
translation of a $(2,\exists)$-QBF $F = \exists X \forall Y D$ into a
disjunctive program $P[F]$ devised by~\citex{eite-gott-95}. The
translation assumes that $D$ is a 3-DNF formula, say
$D =\bigvee_{i=1}^n (l_{i,1}\AND l_{i,2}\AND l_{i,3})$, where
$l_{i,j}$'s are literals over $X\cup Y$. To define $P[F]$ we introduce
mutually distinct fresh atoms $w$, $\overline{x}$, for $x\in X$,
$\overline{y}$, for $y\in Y$, and set 
\begin{align*}%
  {P}[F] = & \{ x \vee \overline{x} \la \mid x\in X\} \cup \{ y \vee
             \overline{y} \la \rsep y \la w\rsep \overline{y} \la w
             \mid y \in Y\} \cup \\
           & \{ w \la l^*_{i,1},l^*_{i,2},l^*_{i,3}\mid 1\leq i \leq
             n\} \cup \{ \bot \la \naf\ w\}
\end{align*}
where $l^*_{i,j} = \naf \overline{x}$ if $l_{i,j}=x$,
$l^*_{i,j} = \naf x$ if $l_{i,j}=\neg x$, $l^*_{i,j} = y$ for
$l_{i,j}=y$, $l^*_{i,j} = \overline{y}$ for $l_{i,j}=\neg y$. It can
be shown that ${P}[F]$ has at least one answer set if and only if $F$
is true. Let us consider the subclass of $(2,\exists)$-QBFs where each
term~$l_{i,1}\AND l_{i,2}\AND l_{i,3}$ in $F$ contains at most one
universally quantified atom from~$Y$. This restriction makes the
$\Sigma^P_2$-complete problem of the validity of a $(2,\exists)$-QBF
$\NP$-complete, only.  Moreover, it is easy to check that under that
restriction, ${P}[F]$ is a dual-normal program. Since, the consistency
problem for dual-normal programs is
$\NP$-complete~\cite{Truszczynski11} as well, dual-normal programs
thus allow here for a straightforward \emph{complexity-sensitive}
reduction with respect to the subclass of the $(2,\exists)$-QBF
problem mentioned above.  \citex{JanhunenEtAl06} proposed another
translation of QBFs into programs that, with slight modifications, is
similarly complexity-sensitive.

\paragraph{Main Contributions}
Our first group of results concerns connections between dual-normal
programs, propositional theories and normal programs. They are
motivated by practical considerations of processing dual-normal
programs. First, we give an efficient translation from dual-normal
programs to \SAT such that the models of the resulting formula encode
the answer sets of the original program. While similar in spirit to
translations to \SAT developed for other classes of programs, our
translation requires additional techniques to correctly deal with the
dual nature of the programs under consideration. Second, in order to
stay within the ASP framework we give a novel translation capable to
express dual-normal programs as normal ones, \emph{and} also vice
versa, in each case producing polynomial-size encodings.  In addition,
this translation allows us to properly extend the class of dual-normal
programs to the novel class of \emph{body-cycle free} programs, a
class for which the consistency problem is still located in \NP.

In the next group of results, we investigate dual-normal programs from 
a different angle: their ability to express concepts modeled by classes
of SE- and UE-models~\cite{Turner01,EiterFPTW13} and, in particular, to express
programs under the notions of equivalence defined in terms of SE- and 
UE-models~\cite{Eiter04TR}. Among others, we show that the classes 
of normal and dual-normal programs are incomparable with respect to 
SE-models, and that dual-normal programs are strictly less expressive 
than normal ones with respect to UE-models. We also present results
concerning the complexity of deciding strong and uniform equivalence
between dual-normal programs. 

\section{Preliminaries}\label{sec:prelimns}

A \emph{rule}~$r$ is an
expression~$\head{r} \la \bodyp{r},\naf \bodyn{r}$, where
$\head{r}=\{a_1\commadots a_l\}$,
$\bodyp{r}=\{a_{l+1}\commadots a_m\}$,
$\bodyn{r}=\{a_{m+1}\commadots a_n\}$, $l$, $m$ and $n$ are
non-negative integers, and $a_i$, $1\leq i\leq n$, are propositional
atoms.  We omit the braces in $H(r)$, $B^+(r)$, and $B^-(r)$ if the
set is a singleton. We occasionally write $\bot$ if
$\head{r}=\emptyset$.  We also use the traditional representation of a
rule as an expression
\begin{align}
  a_1\vee \cdots \vee a_l \leftarrow a_{l+1},\ldots,a_m,\naf
  a_{m+1},\ldots,\naf a_n.
\end{align}
We call $\head{r}$ the \emph{head} of $r$ and
$\body{r}=\{a_{l+1}\commadots a_m,\naf a_{m+1},$ $\ldots, \naf a_n\}$
the \emph{body} of $r$. 
A rule $r$ is \emph{normal} if $\Card{H(r)}\leq 1$, $r$ is \emph{Horn}
if it is normal and $B^-(r)=\emptyset$, $r$ is \emph{dual-Horn} if
$\Card{B^+(r)}\leq 1$ and $B^-(r)=\emptyset$, $r$ is an
\emph{(integrity) constraint} if $\head{r}=\emptyset$, $r$ is
\emph{positive} if $B^-(r)=\emptyset$, and $r$ is \emph{definite} if
$\Card{H(r)} = 1$. If $B^+(r)\cup B^-(r)=\emptyset$, we simply write
$H(r)$ instead of $H(r) \leftarrow \emptyset, \emptyset$.

A \emph{disjunctive logic program} (or simply a \emph{program}) is a
finite set of rules.
We denote the set of atoms occurring in a program~$P$ by $\at(P)$.  We
often lift terminology from rules to programs.  For instance, a
program is \emph{normal} if all its rules are normal. We also identify
the parts of a program $P$ consisting of proper rules as
$P_r= \{ r\in P \mid H(r)\neq\emptyset \}$ and constraints as
$P_c = P \setminus P_r$. In this paper we are particularly interested
in the following class.

\begin{definition}\label{def:dnorm}
  A program~$P$ is called \emph{dual-normal} if each rule~$r$ of $P$
  is either a constraint or $\Card{B^+(r)}\leq 1$. Programs that are
  both normal and dual-normal are called \emph{singular}.\footnote{%
    Singular programs were also considered by~\citex{Janhunen06},
    however under a different name.}
\end{definition}
Note that dual-Horn programs may contain positive constraints with a
single body atom but arbitrary constraints are forbidden in contrast
to dual-normal programs.

Let $P$ be a program and $t$ a fresh atom. We define 
\begin{align*}
  P[t] =& 
          \{H(r)\leftarrow t,\neg B^{-}(r)\mid r\in P, B^+(r)=\emptyset \} \cup  
          \{ r\ \mid r\in P, B^+(r)\neq\emptyset \}.
\end{align*}

This transformation ensures non-empty positive bodies in rules and
turns out to be useful in analyzing the semantics of dual-normal
programs.

An \emph{interpretation} is a set~$I$ of atoms.  An interpretation~$I$
is a \emph{model} of a program~$P$, written $I\models P$, if $I$
\emph{satisfies} each rule $r\in P$, written $I\models r$, that is, if
$(H(r)\cup B^-(r))\cap I\neq \emptyset$ or
$B^+(r) \setminus I \neq \emptyset$.

In the following when we say that a set~$M$ is maximal (minimal) we
refer to inclusion-maximality (inclusion-minimality). A Horn program
either has no models or has a unique least model. Dual-Horn programs
have a dual property.

\begin{proposition}\label{prop:3}
  Let $P$ be dual-Horn. Then $P$ has no models or has a unique maximal
  model. 
\end{proposition}

We will now describe a construction that implies this result and is
also of use in arguments later in the paper.  

  Let us define $E_0=\emptyset$ and, for $i\geq 1$,
  \begin{align*}
    E_{i} = \{b\mid H \leftarrow b \in P[t], \ H\subseteq E_{i-1}\}.
  \end{align*}

Intuitively, the sets~$E_i$ consist of atoms that \emph{must not}
be in any model of $P[t]$ (must be eliminated). The construction 
is dual to that for Horn programs. More precisely, the sets~$E_i$ 
can be alternatively defined as the results of recursively applying 
to $E_0=\emptyset$ the one-step provability operator for the \emph{definite
Horn} program~$P'[t]= \{b \la H\mid H\la b \in P[t]\}$. The following result
summarizes properties of the program~$P[t]$ and sets~$E_i$.

\begin{proposition}\label{prop:4}
  Let $P$ be dual-Horn. Then,
  \begin{enumerate}
  \item $E_0\subseteq E_1 \subseteq \ldots \subseteq \at(P)\cup\{t\}$;
  \item $(\at(P) \cup\{t\})\setminus \bigcup_{i=0}^\infty E_i$ is a
    {maximal} model (over $\at(P) \cup\{t\}$) of $P[t]$;
  \item for every set~$M$ of atoms, $M$ is a model of $P$ if and only
    if $M\cup \{t\}$ is a model of $P[t]$; and
  \item $P$ has a model if and only if $t$ belongs to the maximal
    model (over $\at(P) \cup\{t\}$) of $P[t]$\\ (or, equivalently,
    $t\notin \bigcup_{i=0}^\infty E_i$).
  \end{enumerate}
\end{proposition}

Properties~(3) and~(4) imply Proposition~\ref{prop:3}. The construction can be
implemented to run in linear time by means of the algorithm
by~\citex{DowlingGallier84} for computing the least model of a Horn
program. 

The {\em Gelfond-Lifschitz reduct}~$P^I$ of a program~$P$ {\em
  relative to\/} an interpretation~$I$ is defined as
$P^I = \{ \head{r}\la \bodyp{r} \mid r\in P,I\cap \bodyn{r} =
\emptyset\}$.
Observe that for a dual-normal program~$P$ any reduct of~$P_r$ is
dual-Horn.  An interpretation~$I$ is an \emph{answer set} of a
program~$P$ if $I$ is a minimal model of
$P^I$~\cite{gelf-lifs-91,przy-91}.  The set of all answer sets of a
program $P$ is denoted by $\AS(P)$.

The following well-known characterization of answer sets is often
invoked when considering the complexity of deciding the existence of
answer sets. 

\begin{proposition}\label{prop:1}
  The following statements are equivalent for any program $P$ and any
  set $M$ of atoms: 
  \begin{enumerate}
  \item $M\in \AS(P)$,
  \item $M$ is a model of $P$ and a minimal model of $P_r^M$, and
  \item $M$ is a model of~$P_c$ and $M\in\AS(P_r)$.
  \end{enumerate}
\end{proposition}

This result identifies testing whether an interpretation~$M$ is a
minimal model of $P_r^M$ as the key task in deciding whether $M$ is an
answer set of $P$. For normal programs checking that $M$ is a minimal
model of $P_r^M$ is easy. One just needs to compute the least model of
the Horn program~$P_r^M$ and check whether it matches $M$.  The
general case requires more work. A possible approach is to reduce the
task to that of deciding whether certain programs derived from $P_r^M$
have models.
Specifically, define for a program~$P$ and an atom $m \in \at(P)$
\begin{align*}
  \pmm{P}{M}{m} =
  P_r^M \cup \{ \bot \la b \mid b \in \at(P) \setminus M\} \cup \{ \bot \la m\}.
\end{align*}
With this notation, we can restate Condition~(2) in
Proposition~\ref{prop:1}.

\begin{proposition}\label{prop:2}
  An interpretation~$M$ is an answer set of a program~$P$ if and only
  if $M$ is a model of $P$ and for each~$m\in M$, the
  program~$\pmm{P}{M}{m}$ has no models.
\end{proposition}

Clearly, if a program $P$ is dual-normal, the programs $\pmm{P}{M}{m}$ 
all are dual-Horn. Combining Propositions~\ref{prop:4} and \ref{prop:2} 
yields the following corollary, as well as an efficient algorithm for 
checking whether $M$ is an answer set of $P$.

\begin{corollary}
\label{cor:1}
Let $P$ be a dual-normal program. An 
interpretation $M$ is an answer set of $P$ if and only if $M$ is a model 
of $P$ and for every $m\in M$, $t_m\in \bigcup_{i=0}^\infty E_i$, where $E_i$
are the sets computed based on $\pmm{P}{M}{m}[t_m]$.
\end{corollary}

\section{Translation into \SAT} 

In this section, we encode dual-normal programs as propositional
formulas so that the models of the resulting formulas encode the
answer sets of the original programs. The main idea is to
non-deterministically check for every interpretation whether it is an
answer set of~$P$. In other words, we encode into our formula a guess
of an interpretation and the efficient algorithm described above to
check whether it has models (cf. Corollary~\ref{cor:1}). Note that
the latter part is dual to the Horn encoding
by~\citex{FichteSzeider13}.

Let $P$ be a program and $p=\Card{\at(P)}$. The propositional
variables in our encodings are given by all atoms~$a \in \at(P)$, a
fresh variable~$t$, and fresh variables~$a^i_m$, for
$a\in \at(P)\cup\{t\}$, $m\in \at(P)$, and $0 \leq i \leq p$. We use
the variables~$a^i_m$ and $t^i_m$ to simulate the computation of
$\bigcup_{i=0}^\infty E_i$ based on the program $\pmm{P}{M}{m}[t_m]$,
when testing minimality of an interpretation~$M$ by trying to exclude
$m$ (cf. Corollary~\ref{cor:1}). The superscript~$i$ generates copies
of atoms that represent the set~$E_i$. Moreover, we write
$P \sqcap B$ as a shorthand for $\{r \in P \mid B^+(r)=B\}$
and we write ${\pmm{E_i}{M}{m}}$ to indicate that a set~$E_i$ is
considered with respect to $\pmm{P}{M}{m}[t_m]$ instead of $P[t_m]$.

The following auxiliary formulas simulate, according to
Corollary~\ref{cor:1}, an inductive top-down computation of the
maximal models of~$\pmm{P}{M}{m}[t_m]$, where $M$ is an interpretation
and $m \in M$. Since $\pmm{P}{M}{m}[t_m]$ is dual-Horn the main part
of our first auxiliary formulas is the encoding of the
set~$(\at(P)\cup \{t_m\}) \setminus \bigcup_{i=0}^\infty
{\pmm{E_i}{M}{m}}$
where $m \in M$ and $0 \leq i \leq p$ (cf. Proposition~\ref{prop:4}
Properties~(1) and (2)).

For the initial level~0, the following formula~$F^0_m$ encodes
${\pmm{E_0}{M}{m}}$. That is, it ensures that $m$ does not belong to a
model of $F^0_m$ and all other variables belong to a model of $F^0_m$
if and only if they do for the current interpretation over~$\at(P)$:
\begin{align*}
  \F{0}_m =& \neg m^{0}_m \wedge t^0_m \wedge \bigwedge_{a
             \in \at(P)\setminus\{m\}} (a^{0}_m \leftrightarrow a).
\end{align*}

The next formula encodes the
set~$(\at(P)\cup \{t_m\})\setminus {\pmm{E_i}{M}{m}}$.  In other
words, we ensure that an atom~$a$ does not belong to the model if and
only if there is a rule~$r \in \pmm{P}{M}{m}[t_m]$ where already all
atoms in the head do not belong to the model (according to the
previous step), and analogously for $t^i_m$:
\begin{align*}
  F^{i}_m =& \bigwedge_{a \in \at(P) \setminus \{m\}} \big(
             a^{i}_m \leftrightarrow (a^{i-1}_m \wedge C^i_m(P_r \sqcap\{a\})) \big)
             \wedge
             \big( t^i_m \leftrightarrow ( t^{i-1}_m \wedge
             C^i_m(P_r \sqcap \emptyset))\big)
\end{align*}
\begin{align*}
  (\text{for } 1\leq i \leq p)\quad \text{ where } 
  C^i_m(R) =& \bigwedge_{r \in R} \big( \bigvee_{a\in H(r)}
              a^{i-1}_m \vee \bigvee_{a\in B^-(r)} a\big).
\end{align*}
Note that in $C^i_m(R)$ the heads are evaluated with respect to the
previous level while the negative bodies are evaluated with respect to
the current model candidate, thus simulating the concept of reduct
inherent in $\pmm{P}{M}{m}[t_m]$.

Finally, the following auxiliary formula encodes the condition that an
interpretation satisfies each rule~$r \in P$:
\begin{align*}
  \F{Mod} = \bigwedge_{r \in P} \Big( \bigvee_{a\in H(r)\cup B^-(r)}
  a \vee \bigvee_{a\in B^+(r)} \neg a \Big).
\end{align*}
We now put these formulas together to obtain a formula~$F(P)$
expressing that some interpretation $M\subseteq \at(P)$ is a model of
$P$ and for every atom $a\in M$, atom~$t_a$ does not belong to the
maximal model of~$\pmm{P}{M}{a}[t_a]$:
\begin{align*}
  F(P) = \F{Mod} \wedge \bigwedge_{a \in \at(P)} \Big[a
  \rightarrow \Big( \bigwedge_{i=0}^p F^i_a \wedge \neg t^p_a
  \Big)\Big].
\end{align*}
It is easy to see that the formula~$F(P)$ is of
size~$O(\CCard{P} \cdot \Card{\at(P)}^3)$, where $\CCard{P}$ stands
for the size of $P$, and obviously we can construct it in polynomial
time from $P$. The correctness of the translation is formally stated
in the following result.

\begin{theorem}
  Let $P$ be a dual-normal program.  Then,
  $\AS(P)= \{ M \cap \at(P) \mid M\in\Mod(F(P))\}$, where $\Mod(F)$
  denotes the set of all models of $F$.
\end{theorem}

Our encoding can be improved by means of an explicit encoding of the
induction levels using counters (see~e.g.,~\cite{Janhunen06}). This
allows to reduce the size of the encoding to
$O(\Card{\at(P)} \cdot \CCard{P}\cdot \log \Card{\at(P)})$.

\section{Translation into Normal Programs}
We now provide a polynomial-time translation from programs to programs
that allows us to swap heads with positive bodies. It serves several
purposes. (1)~The translation delivers a normal program when the input
program is dual-normal, and it delivers a dual-normal program when the
input is normal.  Given the complexity results
by~\citex{Truszczynski11}, the existence of such translations is not
surprising. However, the fact that there exists a \emph{single}
bidirectional translation, not tailored to any specific program class,
is interesting.  (2)~When applied to head-cycle free
programs~\cite{Ben-EliyahuD94}, the translation results in programs
that we call \emph{body-cycle} free.  Body-cycle free programs are in
many respects dual to head-cycle free ones.

To proceed, we need one more technical result which provides yet
another characterization of answer sets of programs. It is closely
related to the one given by Corollary~\ref{cor:1} but more convenient
to use when analyzing the translation we give below.  Let $P$ be a
program and $t$ a fresh atom. For every pair of atoms~$x,y$, where
$x\in \at(P)$ and $y\in \at(P)\cup \{t\}$ we introduce a fresh
atom~$y_x$, as an auxiliary atom representing a copy of~$y$ in $P$
with respect to $x$; we clarify the role of these atoms below after
the proof of Proposition \ref{prop:5}.

Moreover, for every set~$Y\subseteq \at(P)\cup\{t\}$, let
$Y_x=\{y_x\mid y\in Y\}$. With this notation in hand, we define
\begin{align*}
  P_x = \{B^{+}_x\leftarrow H_x, \neg B^{-}\mid 
  H\leftarrow B^{+},\neg B^{-}\in P_r[t]\},
\end{align*}
and we write $P^M_x$ for $(P^M)_x$ and $P^M_r$ for $(P^M)_r=(P_r)^M$.

\begin{proposition}\label{prop:5}
  Let $P$ be a program. An interpretation~$M\subseteq \at(P)$ is an
  answer set of $P$ if and only if $M$ is a model of $P$, and for
  every~$x\in M$, $t_x$ belongs to every minimal model of
  $P_x^M\cup\{x_x\}\cup (\at(P)\setminus M)_x$.
\end{proposition}
\begin{proof}
  ($\Leftarrow$) Since $M$ is a model of $P$, $M$ is a model of $P^M$.
  Thus, $M$ is a model of $P_r^M$. By Proposition \ref{prop:1},
  it suffices to show that $M$ is a minimal model of $P_r^M$.

  Let us assume that for some~$N\subset M$, $N\models P_r^M$. Let
  $x\in M\setminus N$.  Finally, let us set $N'=\at(P)\setminus N$.
  We will show that $N_x'$ is a model of $P_x^M$. To this end, let us
  consider a rule~$U_x\la V_x$ in $P_x^M$ such that $U_x\neq \{t_x\}$,
  and assume that $V_x\subseteq N'_x$. It follows that
  $V \subseteq N'$. Since the rule~$V \la U$ belongs to $P_r^M$,
  $N\models P^M$, and $V\cap N=\emptyset$, we have $U\not\subseteq
  N$.
  Thus, $U\cap N'\not=\emptyset$ and so,
  $U_x\cap N_x' \not=\emptyset$. Hence, $N_x'\models U_x\la
  V_x$.
  Next, let us consider a rule~$t_x \la V_x$ in $P_x^M$. Since $V \la$
  is a rule in $P_r^M$ and $N\models P_r^M$, we have
  $V\cap N\not= \emptyset$. Thus, $V\not\subseteq N'$ and so,
  $V_x\not\subseteq N'_x$.  Consequently, $N'_x\models t_x\la V_x$.

  Since $\{x\}\cup (\at(P)\setminus M) \subseteq N'$, it follows that 
  $N'_x\models P_x^M\cup \{x_x\} \cup (\at(P)\setminus M)_x$.
  Since $t\notin N'$, $t_x\notin N_x'$. Thus, there is a minimal model of
  $P_x^M\cup\{x_x\}\cup (\at(P)\setminus M)_x$ that does not contain $t_x$,
  a contradiction (each minimal model of $P_x^M\cup\{x_x\}\cup 
  (\at(P)\setminus M)_x$ contained in $N_x'$ has this property). 

  \noindent
  ($\Rightarrow$) Since $M\in\AS(P)$, $M$ is a model of
  $P$. Let us assume that for some~$x\in M$ and for some minimal model~$N_x'$ 
  of $P_x^M\cup \{x_x\}\cup (\at(P)\setminus M)_x$, $t_x\notin N_x'$. 
  Let us define $N=\at(P)\setminus N'_x$. Since $\{x\} \cup (\at(P)\setminus M)
  \subseteq N'_x$, $N$ is a subset of $M \setminus\{x\}$. Reasoning 
  similarly as before, we can show that $N$ is a model of $P_r^M$. 
  This is a contradiction, as $M$ is  minimal model of $P_r^M$. Thus, 
  the assertion follows by Proposition \ref{prop:1}.
\end{proof}

By Proposition~\ref{prop:5} checking whether $M$ is an answer set of
$P$ requires to verify a certain condition for every $x \in M$.  That
condition could be formulated in terms of atoms in $\at(P)\cup\{t\}$
(by dropping the subscripts $x$ in the atoms of the program $P_x$ and
in the condition). However, if a single normal program is to represent
the condition for all $x\in M$ together, we have to combine the
programs $P_x$. To avoid unwanted interactions, we first have to
standardize the programs apart. This is the reason why we introduce
atoms $y_x$ and use them to define copies of $P_x$ customized to
individual $x$'s.

\newcommand{\Prg}[1]{\ensuremath{P_{\text{#1}}}}

Given a program~$P$ and the customized programs $P_x$, we now describe
the promised translation. To this end, for every atom~$x\in \at(P)$,
we introduce a fresh atom~$\oln{x}$. We set:
\renewcommand{\wh}[1]{\ensuremath{#1_{\mathrm{trans}}}}
\newcommand{\whh}[2]{\ensuremath{#1_{\text{tr,}#2}}}
\begin{align*}
  \Prg{xor} =& \{x \la \naf \oln{x}\rsep
               \oln{x}\la \naf x \mid x\in \at(P)\}\displaybreak[1]\\
  \Prg{aux} =& \{x_x \la \neg\oln{x} \rsep
               y_x \la \neg\oln{x}, \neg y \mid x, y\in \at(P)\}\displaybreak[1]\\
  \Prg{diag} =&  \Prg{xor}\cup \Prg{aux} \nonumber\cup\bigcup_{x\in \at(P)} 
                P_x\displaybreak[1]\\
  \Prg{mod} =& \{\bot \la \naf{H}, B^+,\naf B^-\mid
               H\la B^+,\naf B^-\in P\}\\
  \Prg{true} =& \{\bot \la x, \naf t_x\mid x\in \at(P)\}\displaybreak[1]\\
  \wh{P} =& \Prg{diag}\nonumber\cup \Prg{mod} \cup \Prg{true}
\end{align*}
The following observations are immediate and central:
\begin{enumerate}
\item For a dual-normal program~$P$, $\wh{P}$ is normal.
\item For a normal program~$P$, $\wh{P}$ is dual-normal.
\end{enumerate}

Hence, the following result not only establishes the connection
between the answer sets of $P$ and $\wh{P}$ but also proves that the
transformation encodes dual-normal as normal programs, as desired, and
at the same time, encodes normal programs as dual-normal ones.
Moreover, the transformation can be implemented to run in polynomial
time and so, produces polynomial-size programs.

\begin{samepage}
\begin{theorem}\label{thm:trans2}\nopagebreak
  Let $P$ be a program, $M \subseteq \at(P)$,
  $P'=\bigcup_{x\in M} (P_x^M \cup \{x_x\} \cup(\at(P)\setminus M)_x)$
  and $M_P= M\cup \{\oln{x}\mid x\in \at(P)\setminus M\}$. Then
  $M\in\AS(P)$ if and only if for every minimal model~$N$ of $P'$,
  $M_P\cup N\in\AS(\wh{P})$.  Moreover, every answer set of $\wh{P}$
  is of the form $M_P\cup N$ for $M\subseteq \at(P)$ and a minimal
  model~$N$ of $P'$.
\end{theorem}
\end{samepage}
\begin{proof}
  ($\Rightarrow$) Let $M$ be an answer set of $P$ and let $N$ be any
  minimal model of~$P'$. Since $M$ is a model of $P$ by
  Proposition~\ref{prop:1}, $M_P\cup N$ satisfies all constraints
  in~$\Prg{mod}$. Proposition~\ref{prop:5} implies that for
  every~$x \in M$, $t_x\in N$. Thus, $M_P\cup N$ also satisfies all
  constraints in~$\Prg{true}$.
  To prove that $M_P\cup N\in\AS(\wh{P})$ it remains to show that
  $M_P\cup N\in\AS(\Prg{diag})$ (cf.\ Proposition~\ref{prop:1}).  To
  this end, we observe that, for each $x\in\at(P)$,
  $P_x^{M_P\cup N}=P_x^M$ and thus
  $\Prg{diag}^{M_P\cup N} = \bigcup_{x\in \at(P)} P_x^M \cup M_P \cup
  \bigcup_{x\in M}\big( \{x_x\} \cup (\at(P)\setminus M)_x\big)$.
  Since all rules in $\bigcup_{x\in \at(P)\setminus M} P_x^M$ have a
  nonempty body that is disjoint with $M_P\cup N$, and since $N$ is a
  model of
  $P'=\bigcup_{x\in M} (P_x^M \cup \{x_x\} \cup(\at(P)\setminus
  M)_x)$,
  $M_P\cup N$ is a model of $\Prg{diag}^{M_P\cup N}$. Since $N$ is a
  minimal model of~$P'$, $M_P\cup N$ is a minimal model of
  $\Prg{diag}^{M_P\cup N}$.

  \noindent
  ($\Leftarrow$) Let $N$ be a minimal model of~$P'$ and $M_P\cup N$ an
  answer set of $\wh{P}$.  Clearly, $M_P\cup N$ satisfies the
  constraints in~$\Prg{mod}$ and so, $M$ is a model of $P$.  Let
  $x\in M$. Since $M_P\cup N$ satisfies all constraints
  in~$\Prg{true}$, $t_x\in M_P\cup N$.  Thus, $t_x\in N$. By
  Proposition~\ref{prop:5}, $M$ is an answer set of $P$.

  To prove the second part of the assertion, let us consider an answer
  set~$A$ of $\wh{P}$. Let us define $M= A \cap \at(P)$. Because of
  the rules in~$P^x_{xor}$, $A=M_P\cup N$ for some
  set~$N\subseteq \bigcup_{x\in \at(P)} (\at(P)\cup t)_x$.  By
  Proposition~\ref{prop:1}, $A$ is an answer set of $\Prg{diag}$ that
  is, $A$ is a minimal model of $\Prg{diag}^{A}$. As above, we have
  $\Prg{diag}^{A} = \big(\bigcup_{x\in \at(P)} P_x\big)^M \cup M_P
    \cup \bigcup_{x\in M} \big(\{x_x \} \cup (\at(P)\setminus
    M)_x\big)$
  and conclude that $N$ is a minimal model of~$P'$.
\end{proof}

Our translation allows us to extend the class of dual-normal programs
so that the problem to decide the existence of answer sets remains
within the first level of the polynomial hierarchy. We recall that a
program~$P$ is \emph{head-cycle free $($HCF$)$}~\cite{Ben-EliyahuD94}
if the positive dependency digraph of $P$ has no directed cycle that
contains two atoms belonging to the head of a rule in~$P$.
The \emph{positive dependency digraph} of $P$ has as vertices the
atoms~$\at(P)$ and a directed edge~$(x,y)$ between any two
atoms~$x,y \in \at(P)$ for which there is a rule~$r\in P$ with
$x\in H(r)$ and $y\in B^+(r)$.
It is well known that it is \NP-complete to decide whether a head-cycle
free program has an answer set.  The class of HCF programs arguably is
the most natural class of programs that contains all normal programs
and for which deciding the existence of answer sets is \NP-complete.

We now define a program~$P$ to be \emph{body-cycle free $($BCF$)$} if the 
positive dependency graph of $P$, has no directed cycle that contains two 
atoms belonging to the \emph{positive body} of a rule in $P$. In analogy to HCF 
programs, BCF programs trivially contain the class of dual-normal programs.
Inspecting our translation, yields the following observations:
\begin{enumerate}
\item For a HCF program~$P$, $\wh{P}$ is BCF.
\item For a BCF program~$P$, $\wh{P}$ is HCF.
\end{enumerate}
Since $\wh{P}$ is efficiently obtained from~$P$, the following result
is a direct consequence of Theorem~\ref{thm:trans2} and the fact that
the consistency problem for HCF programs is \NP-complete.

\begin{theorem}\label{thm:bcf}
The problem to decide whether a BCF program~$P$ has an answer set is
\NP-complete.
\end{theorem}

The translation $\wh{P}$ preserves the cycle-freeness of the positive
dependency graph (the positive dependency graph of $P$ is cycle-free
if and only if the positive dependency graph of $\wh{P}$ is
cycle-free). That is essential for our derivation of
Theorem~\ref{thm:bcf}.  However, in general, there is no one-to-one
correspondence between answer sets of $P$ and answer sets of
$\wh{P}$. Thus, as a final result in this section, we provide a slight
adaption of the translation~$\wh{P}$ in which the answer sets of
programs~$P$ and $\wh{P}$ are in a \emph{one-to-one correspondence}.
To this end define,
${P}^\ast = \wh{P} \cup \{ y_x \la t_x \mid x,y\in \at(P)\}$.  Note
that ${P}^\ast$ still turns dual-normal programs to normal programs
and vice versa, but we lose the property that cycle-freeness is
preserved (the new rules may introduce additional cycles in the
positive dependency graph).  Thus, both Theorem~\ref{thm:trans2} and
Theorem~\ref{thm:trans3} are of interest.

\begin{theorem}\label{thm:trans3}
  Let $P$ be a program, $M\subseteq \at(P)$ and $M_P$ as in
  Theorem~\ref{thm:trans2}. Then, $M\in\AS(P)$ if and only if
  $M'=M_P\cup \bigcup_{x\in M} (\at(P) \cup \{t\})_x \in\AS(P^\ast)$.
  Moreover, every answer set of $P^\ast$ is of the form~$M'$ for
  some~$M\subseteq \at(P)$.

\end{theorem}
\nop{
\noindent
\emph{Proof Sketch:}
 First, observe that due to the additional rules the only 
 candidates for models of $\widehat{P}^\ast$ are of the form~$M'$ 
 for some~$M\subseteq \at(P)$.
 Using the insights from Theorem
 \ref{thm:trans2} one can show that only those $M'$ survive  as
 answer sets for which every minimal model of $P_x$ contains $t_x$.
 From this observation, any model of $P_x$ contains $t_x$ but then
 one can show that no~$N\subset M'$ can be model of $(\widehat{P}^\ast)^{M'}$.
}

\section{Expressibility of Dual-Normal Programs}\label{sec:express}

SE-models, originating from the work by~\citex{Turner01}, and
UE-models, proposed by~\citex{eite-fink-03}, characterize strong and
uniform equivalence of programs, respectively. More recently, they
turned out to be useful also for comparing program classes with
respect to their expressivity (see e.g., work
by~\citex{EiterFPTW13}). In what follows, we first recall the main
results from the literature, focusing on disjunctive and normal
programs. Then, we complement these results by characterizations of
collections of SE- and UE-models of dual-normal programs. Finally, we
strengthen existing complexity results.

\newcommand{\SSE}[1]{\ensuremath{\mathcal{S}_{#1}}}
\subsection{SE-models and UE-models}
An \emph{SE-interpretation} is a pair~$(X,Y)$ of sets of atoms such
that $X\subseteq Y$.
We denote by~$\SSE{Z}$ the class~$\SB (X,Y) \SM Y\subseteq Z \SE$ of
all SE-interpretations over~$Z$.
An SE-interpretation~$(X,Y)$ is an \emph{SE-model} of a program~$P$,
written $(X,Y)\mse P$, if $Y\models P$ and $X\models P^Y$.
SE-models of a program $P$ contain, in particular, all 
information needed to identify the answer sets of $P$. Specifically,
$Y$ is an answer set of $P$ if and only if $\langle Y,Y\rangle$ is an 
SE-model of $P$ and for every $X\subset Y$, $\langle X,Y\rangle$ is not.

An SE-model~$(X,Y)$ of a program~$P$ is a \emph{UE-model} of $P$ if
for every SE-model~$(X',Y)$ of $P$ such that $X\subset X'$, $X'=Y$
holds. We write $\SEQ(P)$ ($\UE(P)$) for all SE-interpretations that
are SE-models (UE-models) of a program~$P$.

Programs~$P$ and $Q$ are \emph{equivalent}, denoted by~$P\equiv Q$, if
$P$ and $Q$ have the same answer sets. They are \emph{strongly
  equivalent}, denoted by~$P\equivs Q$, if for every program~$R$,
$P\cup R\equiv Q\cup R$; and \emph{uniformly equivalent}, denoted
$P\equivu Q$, if for every set~$F$ of normal facts,
$P\cup F\equiv Q\cup F$.  The main results concerning these notions
are (1)~$P\equivs Q$ if and only if
$\SEQ(P)=\SEQ(Q)$~\cite{Lifschitz01} and (2)~$P\equivu Q$ if and only
if $\UE(P)=\UE(Q)$~\cite{eite-fink-03}.

We now recall definitions of useful properties of sets of
SE-inter\-pretations~\cite{EiterFPTW13}.  
\begin{definition}
  A set~$\S$ of SE-interpretations is \emph{complete} if
  \begin{enumerate}
  \item $(X,Y)\in \S$ implies $(Y,Y)\in \S$; and
  \item $(X,Y)$, $(Z,Z)\in \S$ and $Y \subseteq Z$ imply
    $(X,Z)\in \S$.
  \end{enumerate}
  Next, $\S$ is \emph{closed under here-intersection} if for
  all~$(X,Y),(X',Y)\in\S$ we have $(X\cap X',Y)\in\S$.
  Finally, $\S$ is \emph{UE-complete} if 
  \begin{enumerate}
  \item $(X,Y)\in\S$ implies $(Y,Y)\in\S$; 
  \item $(X,Y),(Z,Z)\in\S$ and
    $Y\subset Z$ imply that there is $Y'$ such that
    $Y\subseteq Y'\subset Z$ and $(Y',Z)\in\S$; and
  \item $(X,Y), (X',Y)\in \S$ and $X\subset X'$ imply $X'=Y$.
  \end{enumerate}
\end{definition}

The following results are due to~\citex{EiterFPTW13}.  For each
program~$P$, $\SEQ(P)$ is complete. Conversely, for every complete
set~$\S\subseteq\SSE{A}$ there is a program~$P$ with
$\at(P) \subseteq A$ and $\SEQ(P)=\S$.  For each normal program~$P$,
$\SEQ(P)$ is complete and closed under here-intersection. Conversely,
for every set~$\S$ of SE-interpretations over~$A$ that is complete and
closed under here-intersection there is a normal program~$P$ with
$\at(P)\subseteq A$ and $\SEQ(P)=\S$. Next, for every program~$P$,
$\UE(P)$ is UE-complete. Conversely, for every UE-complete
set~$\U \subseteq \SSE{A}$ of SE-interpretations over~$A$ there is a
{normal} program~$P$ such that $\at(P) = A$ and $\U=\UE(P)$.  Hence,
for every disjunctive program~$P$ there exists a normal program~$P'$
with $\UE(P)=\UE(P')$ (however, such $P'$ can be exponentially larger
than $P$~\cite{Eiter04}). Finally, we make use of the following
technical result.
\begin{samepage}
\begin{lemma}\label{lem:se}
  For every SE-interpretation~$(X,Y)$, $(X,Y)\mse A\la B,\naf C$ if
  and only if at least one of the following conditions holds:
  \begin{enumerate}
    \item\label{lem:se:A} $Y\cap C\neq\emptyset$;
    \item\label{lem:se:B} $B\setminus Y\neq\emptyset$;
    \item\label{lem:se:C} $X\cap A\neq\emptyset$;
    \item\label{lem:se:D} $Y\cap A\neq \emptyset$ and $B\setminus X\neq \emptyset$.
  \end{enumerate}
\end{lemma}
\end{samepage}

\paragraph{Properties of Dual-Normal Programs.}
Our results rely on some new classes of sets of SE-interpretations.
First, we introduce sets of SE-interpretations that are closed under
here-union. This is the dual concept to sets closed under here-intersection.
We will use it to characterize the SE-models of dual-normal programs.
To characterize the UE models of dual-normal programs we need an additional,
quite involved, concept of a splittable set. 

\begin{definition}
  A set~$\S$ of SE-interpretations is called 
  \begin{enumerate}
  \item \emph{closed under here-union} if for any~$(X,Y)\in\S$ and
    $(X',Y)\in \S$, also $(X\cup X',Y)\in \S$; 
  \item \emph{splittable} if for every~$Z$ such that $(Z,Z)\in\S$ and
    every~$(X_1,Y_1),\ldots,(X_k,Y_k)\in\S$ such that $Y_i\subseteq Z$
    ($i=1,\ldots, k$), $(X_1\cup\ldots\cup X_k,Z)\in\S$ or
    $(Z',Z)\in\S$ for some~$Z'$, such that
    $X_1 \cup\ldots\cup X_k \subseteq Z'\subset Z$.
  \end{enumerate}
\end{definition}

Neither property implies the other in general. However, for
UE-complete sets of SE-interpretations, splittability implies
closure under here-union.

\begin{proposition}\label{prop:6}
  If a
  UE-complete
  collection~$\S$ of SE-inter\-pre\-tations is splittable, it is closed
  under here-union.
\end{proposition}
\begin{proof}
  Let $(X_1,Z),(X_2,Z)\in\S$. By UE-completeness, $(Z,Z)\in\S$. Thus, if
  $X_1\cup X_2=Z$ then
  $(X_1\cup X_2,Z)\in \S$.  Otherwise, by splittability,
  $X_1\cup X_2\subseteq Z'$ for some $Z'$ such that $Z'\subset Z$ and
  $(Z',Z)\in\S$. Since $X_1\subseteq Z'\subset Z$ and
  $(X_1,Z),(Z',Z)\in \S$, $Z'=X_1$ (by Condition~(3) of UE-completeness).
  Consequently, $X_1\cup X_2=X_1$ and so,
  $(X_1\cup X_2,Z)\in\S$ in this case, too.
\end{proof}

The converse does not hold, that is, for UE-complete sets,
splittability is a strictly stronger concept than closure under
here-union. As an example consider the
set~$\S=\{(b,b),(c,c),(ab,abcd),(cd,abcd),(abcd,abcd)\}$ that is
UE-complete and closed under here-union. This set is \emph{not}
splittable. Indeed, $(abcd,abcd), (b,b),(c,c)\in \S$, yet there is no
$Z'$ such that $\{bc\}\subseteq Z'\subset \{abcd\}$ and
$(Z',abcd)\in\S$.

As announced above, closure under here-union is an essential property
of sets of SE-models of dual-normal programs.

\begin{theorem}\label{thm:se:1}
  For every dual-normal program~$P$, $\SEQ(P)$ is complete and closed under
  here-union.
\end{theorem}
\begin{proof} 
  $\SEQ(P)$ is complete for every program~$P$.  Let
  $(X,Y), (X',Y)\in\SEQ(P)$. We need to show that for every
  rule~$r=A\la B,\naf C$ in $P$, $(X\cup X',Y)\mse r$. To this end,
  let us assume that none of Conditions~(\ref{lem:se:A}),
  (\ref{lem:se:B}), and (\ref{lem:se:C}) of Lemma~\ref{lem:se} holds
  for $(X\cup X',Y)$ and $r$.  Since $X\subseteq X\cup X'$ and
  $X'\subseteq X\cup X'$, none of Conditions~(\ref{lem:se:A}),
  (\ref{lem:se:B}), and (\ref{lem:se:C}) holds for $(X,Y)$ and $r$
  either. Since $(X,Y)\mse r$, Condition~(\ref{lem:se:D}) must hold,
  that is, we have $Y\cap A\neq \emptyset$ and
  $B\setminus X\neq\emptyset$. The same argument applied to $(X',Y)$
  implies that also $B\setminus X'\neq\emptyset$. Since $P$ is
  dual-normal, $B=\{b\}$ and $b\notin X\cup X'$. Thus,
  $B\setminus (X\cup X')\neq\emptyset$ and so,
  Condition~(\ref{lem:se:D}) of Lemma~\ref{lem:se} holds for
  $(X\cup X',Y)$ and $r$.  Consequently, $(X\cup X',Y)\mse r$.
\end{proof}

The conditions of Theorem~\ref{thm:se:1} are not only necessary but
also sufficient.

\begin{theorem}
  \label{thm:dnse}
  For every set~$\S\subseteq\S_A$ of SE-interpretations 
  that is complete and closed under here-union, there exists a
  dual-normal program~$P$ with $\at(P) \subseteq A$ and $\SEQ(P)=\S$.
\end{theorem}
\begin{proof}
  Let $Z$ be a set of atoms, $\S \subseteq \S_{Z}$ a set of
  SE-interpretations that is complete and closed under here-union, and
  $\Y = \{Y\colon (X,Y)\in \S\}$.  Consider $\hY\subseteq Z$ such
  that $(\hY,\hY)\notin \S$. Since $\S$ is complete, for every
  $Y\in\Y$, $(Y,Y)\in\S$. Thus, for every~$Y\in\Y$, $Y\neq\hY$. We
  define
  \begin{align*}
    \Y'=\{Y\in\Y\colon Y\subseteq \hY\} \ \text{ and }\
    \Y''=\{Y\in \Y\colon Y\setminus \hY\neq\emptyset\}.
  \end{align*}

  Clearly, $\Y''\cap\Y'=\emptyset$ and $\Y'\cup \Y''=\Y$. For
  each~$Y\in\Y'$, we select an element~$b_Y\in \hY\setminus Y$ (it is
  possible, as $Y\neq \hY$).  Similarly, for each~$Y\in\Y''$, we
  select an element~$c_Y\in Y\setminus \hY$. We set
  $B_{\hY}= \{b_Y\colon Y\in \Y'\}$ and
  $C_{\hY}=\{c_Y\colon Y\in\Y''\}$, and we define
  \begin{align*}
    r_{\hY}=\quad\la B_{\hY},\naf C_{\hY}.
  \end{align*}

  We note that for every~$(X,Y) \in \S$, $(X,Y)\mse r_{\hY}$.  Indeed,
  if $Y\in\Y'$, then $b_Y\in B_{\hY}\setminus Y$ and so,
  Condition~(\ref{lem:se:B}) of Lemma~\ref{lem:se} holds. Otherwise,
  $Y\in\Y''$ and $c_Y\in C_{\hY}\cap Y$. Thus,
  Condition~(\ref{lem:se:A}) of that lemma holds. On the other hand,
  $(\hY,\hY)\not\mse r_{\hY}$. Indeed, $C_{\hY}\cap \hY=\emptyset$ and
  $B_{\hY} \subseteq \hY$, so neither Condition~(\ref{lem:se:A}) nor
  Condition~(\ref{lem:se:B}) holds.  Moreover, neither
  Condition~(\ref{lem:se:C}) nor Condition~(\ref{lem:se:D}) holds, as
  $r_{\hY}$ is a constraint.

  Next, let us consider $(\hX,\hY)\notin\S$, where $\hY\in\Y$, and let
  us define $\X = \{X\colon (X,\hY)\in \S\}$. We set
  \begin{align*}
    \X'= \{X\in \X\colon X\subseteq \hX\}\ \ \mbox{and}\ \ \X''=\{X\in
    \X\colon X\setminus\hX\neq\emptyset\}.
  \end{align*}
  
  If $\X'\neq\emptyset$, let $X_0=\bigcup\X'$. Since $\S$ is closed
  under here-union, $X_0$ is a \emph{proper} subset of $X$. We select
  an arbitrary element~$b \in \hX\setminus X_0$ and define
  $B=\{b\}$. Otherwise, we define $B=\emptyset$.

  If $\X''\neq\emptyset$, for each~$X\in \X''$, we
  select~$a_X\in X\setminus \hX$, and we define
  $A=\{a_X\colon X\in\X''\}$.  Otherwise, we select any
  element~$a\in \hY\setminus\hX$ and define $A=\{a\}$. We note that by
  construction, $A\subseteq \hY$.

  Next, we define
  \begin{align*}
    \Z=\{Y\in\Y\setminus\{\hY\}\colon Y\setminus \hY\neq\emptyset\}.
  \end{align*}

  For each~$Y\in \Z$, we select~$c_Y\in Y\setminus\hY$ and
  set~$C=\{c_Y\colon Y\in \Y'\}$.

  Finally, we define a rule~$r_{(\hX,\hY)}$ as
  \begin{align*}
    r_{(\hX,\hY)} = A \la B,\naf C.
  \end{align*}

  It is easy to see that $(\hX,\hY)\not\mse r_{(\hX,\hY)}$. Indeed, by
  construction, $\hY\cap C=\emptyset$, $B\subseteq \hX\subseteq\hY$,
  and $A \cap \hX=\emptyset$. The second condition implies that
  $B\setminus\hY=\emptyset$ and $B\setminus\hX=\emptyset$. Thus, none
  of the Conditions~(1)--(4) of Lemma~\ref{lem:se} holds.

  We will show that for every~$(X,Y)\in\S$, $(X,Y)\mse r_{(\hX,\hY)}$.
  First, assume that $Y\setminus\hY\neq\emptyset$. It follows that
  $c_Y\in C\cap Y$ and so, $C\cap Y\neq\emptyset$. Thus,
  $(X,Y)\mse r_{(\hX,\hY)}$ by Condition~(\ref{lem:se:A}).

  Assume that $Y\subseteq\hY$. Since $(X,Y)\in \S$ and
  $(\hY,\hY)\in\S$, $(X,\hY)\in \S$. Thus, $X\in \X$. If
  $X\setminus\hX \neq\emptyset$, then $\X\in\X''$ and so,
  $X\cap A\neq\emptyset$.  Consequently, $(X,Y)\mse r_{(\hX,\hY)}$ by
  Condition~(\ref{lem:se:C}). Otherwise, $X\in \X'$ and $B=\{b\}$, for
  some~$b\in \hX\setminus X_0$. In particular,
  $B\setminus X\neq\emptyset$.  Since $(X,Y)\in \S$, $(Y,Y) \in \S$
  and so, $(Y,\hY)\in \S$.  Consequently, $Y\in \X$. If $Y\in\X''$,
  then $Y\cap A \neq\emptyset$ and $(X,Y)\mse r_{(\hX,\hY)}$ by
  Condition~(\ref{lem:se:D}). If $Y\in \X'$, then
  $b\in \hX\setminus Y$ and so, $B\setminus Y\neq\emptyset$.  Thus,
  $(X,Y)\mse r_{(\hX,\hY)}$ by Condition~(\ref{lem:se:B}).

  Let $P$ consist of all rules~$r_{\hY}$, where $\hY\subseteq Z$ and
  $Y\notin Y$ and of all rules~$r_{(\hX,\hY)}$ such that
  $\hX,\hY\subseteq Z$, $\hX\subseteq\hY$ and
  $(\hX,\hY)\notin\S$. Clearly, $\S\subseteq SE(P)$.  Let
  $(\hX,\hY)\notin\S$. If $\hY\notin\Y$, then
  $(\hY,\hY)\not\mse r_{\hY}$.  Thus, $(\hX,\hY)\notin SE(P)$. If
  $\hY\in\Y$, then $(\hX,\hY)\not\mse r_{(\hX,\hY)}$. Thus,
  $(\hX,\hY)\notin SE(P)$. It follows that $\SEQ(P)=\S$.
\end{proof}

Thus the two theorems together provide a complete characterization of
collections of SE-interpretations that can arise as collections of
SE-models of dual-normal programs.

We now turn to the corresponding results for sets of UE-models of
dual-normal programs. The key role here is played by the notion of
splittability.

\begin{theorem}\label{thm:ue:1}
  For every dual-normal program~$P$, $\UE(P)$ is UE-complete and
  splittable.
\end{theorem}
\begin{proof}
  The set~$\UE(P)$ is UE-complete for every program~$P$. Thus, we only
  need to show splittability. Toward this end, let
  $(X_1,Y_1),\ldots,(X_k,Y_k),(Z,Z) \in \UE(P)$, where
  $Y_i\subseteq Z$, for every $i=1,\ldots,k$.
  Since,$(X_1, Y_1),\ldots,(X_k,Y_k),(Z,Z)\in\SEQ(P)$, it follows that
  $(X_1,Z),\ldots,(X_k,Z) \in \SEQ(P)$ (by the second condition of
  completeness). Since $\SEQ(P)$ is closed under here-union,
  $(X_1\cup\ldots\cup X_k,Z)\in \SEQ(P)$. If $(X_1\cup\ldots\cup X_k,Z)
  \in \UE(P)$ we are done. Otherwise, $X_1\cup\ldots\cup X_k\subset Z$
  (since $(Z,Z)\in\UE(P)$) and, by the definition of UE-models and
  finiteness of $P$, there is $Z'$ such that $X_1\cup\ldots \cup X_k\subset
  Z'\subset Z$ such that $(Z',Z)\in\UE(P)$.
\end{proof}

As before, the conditions are also sufficient.

\begin{theorem}\label{thm:dnue} 
  For every set $\U\subseteq\S_A$ of SE-interpretations that is
  UE-complete and splittable, there is a dual-normal program~$P$ with
  $\at(P) \subseteq A$ such that $\UE(P)=\U$.
\end{theorem}
\begin{proof}
  For every~$Z$ such that $(Z,Z)\in\U$, we define
  \begin{align*}
    \U_Z=\{X\colon (X,Y)\in\U, \text { for some } Y\subseteq Z \}.
  \end{align*}
  and we denote by $cl(\U_Z)$ the closure of $\U_Z$ under union.
  Finally, we define the \emph{SE-closure} $\overline{\U}$ of $\U$ by
  setting
  \begin{align*}
    \overline{\U}=\{(X,Z)\colon X\in cl(\U_Z)\}.
  \end{align*}

  We note that if $(X,Z)\in \overline{\U}$, then $X\in cl(\U_Z)$.
  Thus, $\U_Z$ is defined, that is, $(Z,Z)\in\U$. Consequently,
  $Z\in cl(\U_Z)$ and $(Z,Z)\in\overline{\U}$.

  Next, assume that $(X,Y)\in\overline{\U}$, $(Z,Z)\in\overline{\U}$,
  and $Y\subset Z$. It follows that $X\in cl(\U_Y)$. Thus, there are
  sets~$X_1,\ldots, X_k$ such that $X=\bigcup_{i=1}^n X_i$ and
  $X_i\in\U_Y$, for every $i=1,\ldots, k$.  Let us consider any such
  set~$X_i$. By definition, there is a set~$Y'$ such that
  $(X_i,Y')\in \U$ and $Y'\subseteq Y$. Since $Y \subseteq Z$,
  $Y'\subseteq Z$. It follows that $X_i\in\U_Z$.  Thus,
  $X_1, \ldots,X_k\in \U_Z$. Consequently, $X\in cl(\U_Z)$ and
  $(X,Z) \in \overline{\U}$.

  Thus, $\overline{\U}$ is complete and, by the construction, closed
  under here-unions. It follows that there is a dual-normal
  program~$P$ such that $\SEQ(P)=\overline{\U}$. We will show that
  $\UE(P)=\U$.

  First, let $(X,Y) \in\U$. It follows that $X\in\U_Y$. Thus,
  $X\in cl(\U_Y)$ and $(X,Y)\in \overline{\U}$. Consequently,
  $(X,Y)\in\SEQ(P)$. Let us assume that for some $(X',Y)\in\SEQ(P)$,
  $X\subset X'\subset Y$. Since $(X',Y)\in \SEQ(P)$,
  $(X',Y)\in \overline{\U}$ and so, $X'\in cl(\U_Y)$. Thus,
  $X'=X_1\cup\ldots\cup X_k$, where $X_1,\ldots,X_k\in\U_Y$ or,
  equivalently, $(X_1,Y),\ldots,(X_k,Y)\in\U$. Since $X'\subset Y$, it
  follows by splittability that there is $(Y',Y)\in\U$ such that
  $Y'\subset Y$ and $X_1\cup\ldots\cup X_k\subseteq Y'$. Since
  $(X_1,Y)\in\U$ and $X_1\subseteq Y'\subset Y$, it follows that
  $X_1=Y'$. Consequently, $X'=X_1\cup\ldots\cup X_k=Y'$.  Thus,
  $(X',Y)\in\U$, a contradiction. It follows that $(X,Y)\in\UE(P)$.

  Conversely, let $(X,Y)\in\UE(P)$. It follows that $(X,Y)\in\SEQ(P)$
  and, since $\SEQ(P)=\overline{\U}$, $(X,Y)\in\overline{\U}$. By the
  definition, $X\in cl(\U_Y)$. Since $\U_Y$ is defined,
  $(Y,Y)\in\U$. Thus, if $X=Y$, the assertion follows. Otherwise,
  $X\subset Y$. In this case, we reason as follows. Since
  $X\in cl(\U_Y)$, as before we have $X=X_1\cup\ldots\cup X_k$, for
  some sets~$X_i$, $1\leq i\leq k$, such that $(X_i,Y)\in \U$.  By
  splittability, there is $Y'$ such that
  $X_1\cup\ldots\cup X_k\subseteq Y'$, $Y'\subset Y$ and
  $(Y',Y)\in\U$. Again as before, we obtain that $X_1= Y'$ and so,
  $X=X_1\cup\ldots\cup X_k=Y'$. Thus, $(X,Y)\in\U$.
\end{proof}

We briefly discuss some implications of our results.  Let
\begin{align*}
  P =& \{ a \vee b\rsep \bot \la \naf c\rsep c\la a,b\rsep a\la c\rsep
       b\la c\}.
\end{align*}
Then $\SEQ(P)=\{(abc,abc),(a,abc),(b,abc)\}$ and it is neither closed
under here-union nor under here-intersection.  Thus, for $P$ there are
no strongly equivalent programs in the classes of normal and
dual-normal programs.  Moreover, $\UE(P)$ is not closed under
here-union and so, not splittable (Proposition~\ref{prop:6}).
Therefore there is no dual-normal program~$P'$ such that $P\equivu P'$
(such a normal $P'$ exists, however).  Now let us consider the normal
program $Q=P\setminus \{a \vee b\}$. We have
$\SEQ(Q)=\SEQ(P)\cup\{(\emptyset,abc)\}$.  Since $\SEQ(Q)$ is not
closed under here-union, there is no dual-normal program strongly
equivalent to $Q$.  Finally, consider the dual-normal program
$R=P\setminus \{ c\la a,b\}$. We have
$\SEQ(R)=\SEQ(P)\cup \{(ab,abc)\}$.  Since $\SEQ(R)$ is not closed
under here-intersection, there is no normal program strongly
equivalent to $R$.

\subsection{Complexity}
We complement the following known results~\cite{Eiter04TR}: Checking
strong equivalence between programs is $\coNP$-complete; tractability
is only known for the case when both programs are Horn. Checking
uniform equivalence between programs is $\PiP{2}$-complete. If one of
the programs is normal, then the problem is $\coNP$-complete.

\begin{theorem}\label{the:comp-strong}
  Checking strong equivalence between singular programs remains
  $\coNP$-hard.
\end{theorem}
\begin{proof}
  Take the standard reduction from \UNSAT (as e.g.\ used
  by~\citex{PearceTW09}) and let
  $F=\bigwedge_{i=1}^n (l_{i1} \vee l_{i2} \vee l_{i3})$. Define the
  singular program %
  \begin{align*}
    P[F] =& \{ v\la \naf \bar{v}\rsep \bar{v}\la\naf v\rsep
             \la v,\naf{v} \mid v\in \at(F)\} \cup
         \{ \la \naf l^*_{i1}, \naf l^*_{i2} , \naf l^*_{i3} \mid 1\leq i \leq n
            \}
  \end{align*}
  where $l^*=l$ for positive literals and $l^*=\bar{v}$ for negative
  ones.  One can show that $F$ is a positive instance of \UNSAT if and
  only if $P[F]\equivs \{a \la\rsep \la a\}$.  Since the reduction
  works in polynomial time, $\coNP$-hardness follows.
\end{proof}

\begin{theorem}
  Checking uniform equivalence between dual-normal programs is
  $\coNP$-complete. Hardness holds even in the case the programs are
  singular.
\end{theorem}
\begin{proof}
  For membership, consider the following algorithm for the
  complementary problem.  We guess $(X,Y)$ and check whether
  $(X,Y)\in\UE(P)\setminus \UE(Q)$ or
  $(X,Y)\in\UE(Q)\setminus \UE(P)$.  Checking whether $(X,Y)\in\UE(P)$
  can be done efficiently: First check $(Y,Y)\in\UE(P)$ which reduces
  to classical model checking.  If the test fails or $X=Y$ we are
  done. Otherwise, we compute for each~$y\in Y\setminus X$ the maximal
  models of the dual-Horn theories
  \begin{align*}
    P^Y \cup X \cup \{ \la z \mid z\in \At\setminus Y \}
    \cup \{\la y\}.
  \end{align*}
  This can be done in polynomial time, too. If all maximal models are
  equal to $X$, we return true; otherwise false.  For hardness, one
  can employ the reduction used in the proof of Theorem~6.6
  in~\cite{Eiter04TR}.
\end{proof}

\section{Conclusions}
We studied properties of dual-normal programs, the ``forgotten'' class
of disjunctive programs, for which deciding the existence of answer
sets remains \NP-complete. We provided translations of dual-normal
programs to propositional theories and to normal programs, and
characterizations of sets of SE-interpretations that arise as sets of
SE- and UE-models of dual-normal programs. We also established the
\coNP-completeness of deciding strong and uniform equivalence between
dual-normal programs, showing hardness even under additional syntactic
restrictions.

Our paper raises several interesting issues for future work. First,
the BCF programs that we introduced as a generalization of dual-normal
programs deserve further study because of their duality to HCF
programs, and good computational properties (\NP-completeness of
deciding existence of answer sets).  We believe that BCF programs
provide a promising class to encode certain problems, since they also
allow certain conjunctions in the positive body.  Recall that the
operation of \emph{shifting} transforms HCF programs into normal ones
while preserving the answer sets~\cite{Ben-EliyahuD94}. An analog of
shifting for BCF programs would introduce negations in the heads of
the rules. Thus, we plan to explore shifting within the broader
setting of Lifschitz-Woo programs~\cite{lw92}.
On the other hand, singular programs, another class of programs we
introduced, deserve attention due to their simplicity --- they are
both normal and dual-normal. As concerns dual-normal programs
themselves, the key question is to establish whether more concise
translations to \SAT and normal programs are possible, as such
translations may lead to effective ways of computing answer sets.

\bibliographystyle{named}

\end{document}